\documentclass[a4paper]{amsart}
\pdfoutput=1

\usepackage[utf8]{inputenc}
\usepackage[colorlinks=true,
linkcolor=blue,
citecolor=blue,
filecolor=blue, 
urlcolor=blue   
]{hyperref}
\usepackage{tensor}
\usepackage{amsmath}
\usepackage{amssymb}
\usepackage{amsthm}
\usepackage{amsopn}
\usepackage{mathtools}
\usepackage{amsfonts}
\usepackage{graphicx}
\usepackage{cite}
\usepackage{subcaption}
\usepackage{color}
\usepackage{mathrsfs}
\usepackage[normalem]{ulem}

\usepackage{tikz}
\usetikzlibrary{decorations.pathmorphing}
\usetikzlibrary{arrows.meta}

\usepackage{xpatch}
\xpatchcmd{\proof}{\itshape}{\normalfont\proofnamefont}{}{}
\newcommand{\proofnamefont}{\sc}

\usepackage{newunicodechar}
\newunicodechar{∂}{\partial}
\newunicodechar{∀}{\forall}
\newunicodechar{∈}{\in}
\newunicodechar{×}{\times}
\newunicodechar{α}{\alpha}
\newunicodechar{β}{\beta}
\newunicodechar{γ}{\gamma}
\newunicodechar{δ}{\delta}
\newunicodechar{ε}{\epsilon}
\newunicodechar{η}{\eta}
\newunicodechar{ι}{\iota}
\newunicodechar{φ}{\varphi}
\newunicodechar{ϕ}{\phi}
\newunicodechar{λ}{\lambda}
\newunicodechar{ψ}{\psi}
\newunicodechar{ρ}{\rho}
\newunicodechar{ω}{\omega}
\newunicodechar{τ}{\tau}
\newunicodechar{θ}{\theta}
\newunicodechar{μ}{\mu}
\newunicodechar{ν}{\nu}
\newunicodechar{π}{\pi}
\newunicodechar{σ}{\sigma}
\newunicodechar{ξ}{\xi}
\newunicodechar{Γ}{\Gamma}
\newunicodechar{Ψ}{\Psi}
\newunicodechar{Ω}{\Omega}
\newunicodechar{Σ}{\Sigma}
\newunicodechar{Θ}{\Theta}
\newunicodechar{Λ}{\Lambda}
\newunicodechar{Δ}{\Delta}

\newcommand{\ie}{{i.e.\ }}

\DeclareMathOperator{\diam}{diam}

\newcommand{\hyp}{{\mycal S}}
\newcommand{\mcM}{\mathscr{M}}

\newcommand{\R}{\ensuremath{\mathbb{R}}}
\newcommand{\N}{\ensuremath{\mathbb{N}}}

\newcommand{\D}{\mathrm{d}}
\newcommand{\dx}{\mathrm{d}x}

\newcommand{\ds}{\mathrm{d}s}
\newcommand{\dt}{\mathrm{d}t}

\newcommand{\Ni}[1]{\mathcal{N}\indices{#1}}

\newtheoremstyle{mytheoremstyle} 
    {\medskipamount}                    
    {\medskipamount}                    
    {\itshape}                   
    {}                           
    {\scshape}                   
    {.}                          
    {5pt plus 1pt minus 1pt}                       
    {}  

\theoremstyle{mytheoremstyle}

\newtheorem{theorem}{Theorem}[section]
\newtheorem{lemma}[theorem]{Lemma}
\newtheorem{proposition}[theorem]{Proposition}

\newtheoremstyle{mydefremarkstyle} 
    {\medskipamount}                    
    {\medskipamount}                    
    {}                   
    {}                           
    {\scshape}                   
    {.}                          
    {5pt plus 1pt minus 1pt}                       
    {}  

\theoremstyle{mydefremarkstyle}

\newtheorem{remark}[theorem]{Remark}

\newtheorem{definition}[theorem]{Definition}

\newcounter{mnotecount}[section]

\renewcommand{\themnotecount}{\thesection.\arabic{mnotecount}}
\newcommand{\mnote}[1]
{\protect{\stepcounter{mnotecount}}$^{\mbox{\footnotesize
$
\bullet$\themnotecount}}$ \marginpar{
\raggedright\tiny\em
$\!\!\!\!\!\!\,\bullet$\themnotecount: #1} }

\definecolor{cadmiumgreen}{rgb}{0.0, 0.42, 0.24}

\DeclareFontFamily{OT1}{rsfs}{}
\DeclareFontShape{OT1}{rsfs}{m}{n}{ <-7> rsfs5 <7-10> rsfs7 <10-> rsfs10}{}
\DeclareMathAlphabet{\mycal}{OT1}{rsfs}{m}{n}

{\catcode `\@=11 \global\let\AddToReset=\@addtoreset}
\AddToReset{equation}{section}

{\catcode `\@=11 \global\let\AddToReset=\@addtoreset}
\AddToReset{figure}{section}

{\catcode `\@=11 \global\let\AddToReset=\@addtoreset}
\AddToReset{table}{section}

\renewcommand{\mcM}{M}

\begin{document}

\title{The annoying  null boundaries }

\author[P.T. Chru\'sciel]{Piotr T.~Chru\'sciel}

\address{Piotr	T.~Chru\'sciel, Faculty of Physics and Erwin Schr\"odinger Institute, University of Vienna, Boltzmanngasse 5, A1090 Wien, Austria}
\email{piotr.Chru\'sciel@univie.ac.at} \urladdr{http://homepage.univie.ac.at/piotr.Chru\'sciel/}

\author[P. Klinger]{Paul Klinger}

\address{Paul Klinger, Faculty of Physics and Erwin Schr\"odinger Institute, University of Vienna, Boltzmanngasse 5, A1090 Wien, Austria}
\email{paul.klinger@univie.ac.at}

\thanks{Preprint UWThPh-2017-36}
%

\begin{abstract}
We consider a class of globally hyperbolic space-times with ``expanding singularities''. Under suitable assumptions we show that no $C^0$-extensions across a compact boundary  exist, while the boundary must be null wherever differentiable (which is almost everywhere) in the non-compact case.
\end{abstract}

\maketitle

\tableofcontents

\section{Introduction}

\medskip

One of the major open questions in mathematical general relativity is the behavior of globally hyperbolic space-times when singular boundaries are approached. In particular the question of extendibility of the metric across such boundaries lies at the heart of the ``cosmic censorship conjecture''~\cite{PenroseSCC}, compare~\cite{ChrCM,SCC}.

In a recent important paper, Sbierski~\cite{SbierskiSchwarzschild} has established $C^0$-inextendibility of the Kruskal-Szekeres extension of the Schwarzschild metric. His proof  makes uses of the $SO(2)\times \R$ symmetry of the metric, which renders the argument unsuitable in situations where no isometries exist. (See also~\cite{SbierskiLingGalloway}.)

It is of interest to enquire whether  Sbierski's analysis can be adapted to more general spacetimes of interest, without isometries.
 The object of this note is to point out   a class of space-times where $C^0$-extendibility can only happen across null boundaries, if at all. Indeed, let us consider a globally hyperbolic space-time with a differentiable Cauchy time function $t$ covering $(0,\infty)$; we will be interested in possible extensions of $\mcM$ towards the past,  see Definition~\ref{def:fpextension} below. The function $t$ determines a topological splitting $\mcM=(0,\infty)\times \hyp$, where the  slices   $\{\tau\}\times \hyp$ are the level sets of $t$; here  one travels from a slice  $\{t_1\}\times\hyp$ to $\{t_2\}\times\hyp$ by following the integral curves of $\nabla t$.

\begin{definition}\label{def:expanding_global}
We shall say that a \underline{globally hyperbolic}
space-time $(\mcM,g)$, with a time function $t$ as just described, contains a \emph{globally  expanding singularity towards the past} if for every open set $A\subset \hyp$ there exists a sequence $t_i$ decreasing to zero such that the (Riemannian) diameter of $\{t_i\}\times A$ within $\{t_i\}\times \hyp$  tends to infinity as $t_i\to 0$.
\end{definition}

As discussed in more detail in Section~\ref{s21X16.1} below, the space-times constructed in~\cite{Chrusciel2015,Klinger2015}, as well as Gowdy, and various other
``Asymptotically Velocity Term Dominated'' (AVTD) spacetimes obtained via Fuchsian methods are of this type, and the theorems below apply.

All the boundaries that we consider in this work will be achronal, hence differentiable almost everywhere by standard arguments (cf. the beginning of Section~\ref{s2-X17.11}).

We have:

\begin{theorem}
  \label{T12VI17.1-}
Suppose that $(\mcM,g)$ contains  a globally expanding singularity towards the past.
 Then in every continuous past extension
  of $\mcM$ the boundary of $\partial \mcM$ is null at all its differentiability points.
\end{theorem}

Theorem \ref{T12VI17.1-} is an immediate consequence of Theorem~\ref{T12VI17.1} below, which is proved under a condition weaker, but somewhat more involved, than that of Definition~\ref{def:expanding_global}.

As such, it also holds:

\begin{proposition}
	Under the hypotheses of Theorem~\ref{T12VI17.1}, every spacelike hypersurface in the extension which intersects the boundary $\partial \iota(\mcM)$ of the image $\iota(\mcM)$ also intersects $\iota(\mcM)$ itself.
\end{proposition}

Some further results in the same spirit  concerning general extensions are also established, see Theorems~\ref{T16VI17.1} and \ref{T27VI17.1} below.

We suspect that null extensions cannot occur either under the hypotheses above, and thus such space-times are inextendible, but we have not been able to establish this in general. However we have:

\begin{proposition}
  \label{P30X17}
Under the hypotheses of Theorem~\ref{T12VI17.1-}, $(M,g)$ has no  extensions with a continuous metric and a compact boundary.
\end{proposition}

The proof of Proposition~\ref{P30X17} can be found in Section \ref{Sproofs}.

Recall that electrovacuum space-times with compact Cauchy horizons have been studied in~\cite{VinceJimcompactCauchy,VinceJimcompactCauchyCMP}, and that,  in space-times in which the metric is $C^3$-extendible~\cite{MinguzziArea,LarssonArea}, for such metrics compactness of the horizon implies its differentiability. We emphasise that our arguments do not need such results.

\medskip
It should be recognised that our definition of expanding space-times is tied to the choice of a time-function $t$, and large deformations of a good time function, if one exists, will not preserve the condition. In particular   it might be very difficult to determine whether or not a given space-time, presented in a coordinate system where the  conditions of Definition~\ref{def:expanding_global} are \emph{not} met, admits a time-function which will satisfy the conditions.  But our results here give some geometric meaning to the notion: space-times extendible through a compact Cauchy horizon, or past-extendible through a spacelike boundary, will not be expanding in the sense of Definition~\ref{def:expanding_global} no matter what time function is used.

\section{Conventions and definitions}

\medskip
We use the standard definition of the Riemannian diameter of a set $A$ in the set $B\supseteq A$, where both $A$ and $B$ are subsets of a Riemannian manifold $(M,g)$
\begin{equation}\label{riem_def}
\diam(A,B):=\sup_{x, y\in A}\inf_{\substack{γ:[0,1]\to B\\ γ(0)=x\,,\thinspace γ(1)=y}}\int_{0}^{1}|\dot{γ}(s)|_g \ds \,.
\end{equation}
Our remaining definitions follow Sbierski~\cite{SbierskiSchwarzschild}, in particular:

\begin{definition}
  \label{d18VI17.1}
 A \emph{$C^0$-extension} of a spacetime $(M, g)$, where $M$ is a smooth manifold and $g$ a Lorentzian metric, is a spacetime $(\tilde M, \tilde{g})$ of the same dimension, with $\tilde M$ again a smooth manifold and $\tilde{g}$ a continuous Lorentzian metric, together with an isometric embedding $ι: M \to \tilde M$ such that $ι(M)$ is a proper subset of $\tilde M$.
\end{definition}

The timelike futures and pasts $I^\pm$ will be defined using piecewise smooth timelike curves in both $M$ and $\tilde M$.

We define the future/past boundary of a spacetime as in~\cite{GallowayLing}, namely:

\begin{definition}\label{def:fpboundary}
The \emph{future, respectively past, boundary of $\mcM$}, denoted $∂^+ ι(M)$, respectively $∂^- ι(M)$, is the set of points $p\in ∂ι(\mcM)$ such that there exists a future/past directed timelike curve $γ: [0,1]\to \tilde{\mcM}$ with $γ([0,1))\subset ι(\mcM)$ and $γ(1)=p$.
\end{definition}

\begin{definition}\label{def:fpextension}
An extension will be called   \emph{past}, respectively \emph{future}, if $∂^+ι(M)=\emptyset$, respectively $∂^-ι(M)=\emptyset$.
\end{definition}

We will also need a definition closely related to, but somewhat weaker than  Definition~\ref{def:expanding_global}:

\begin{definition}\label{def:expanding}
	We shall say that a globally hyperbolic space-time $(\mcM,g)$ contains an \emph{expanding singularity towards the past} if there exists a time function $t$ ranging over $(0,\infty)$ and a real number $t_L>0$ such that for all pairs of points $p\in \{t<t_L\}$ and $q∈ I^-(p,M)$ there exists a sequence $t_i$ decreasing to zero such that the diameter of $I^-(q,M)\cap\{t=t_i\}$ in $I^-(p,M)\cap \{t=t_i\}$ tends to infinity as $t_i\to 0$.
\end{definition}

This is illustrated in Figure~\ref{fig:expanding}.

	\begin{remark}
We note that a globally expanding singularity in the sense of Definition~\ref{def:expanding_global} is expanding in the sense of Definition~\ref{def:expanding}. For this, let us identify the spacetime $M$ with $\R\times \hyp$ by flowing along the gradient $\nabla t$ of $t$. For a subset $\Omega$ of $\hyp$ set $\Omega(t)=\{t\} \times \Omega$,  and let
$$
 \Omega_i(t_i) := I^-(q,M)\cap\{t=t_i\}
 \,.
$$
We then have, for $j>i$,
$$
 \Omega_i(t_j)  \subseteq I^-(q,M)\cap\{t=t_j\} \equiv \Omega_j(t_j)
  \,.
$$
By its definition, given in \eqref{riem_def}, the diameter appearing in Definition \ref{def:expanding}, namely
$$
 \diam(\Omega_i(t_i), I^-(p,M)\cap \{t=t_i\})
 \,,
$$
is bounded from below by the diameter of $\Omega_i(t_i)$ in the whole level set,  $\diam(\Omega_i(t_i),\{t=t_i\})$.
 By Definition \ref{def:expanding_global} of a globally expanding singularity the diameter of  $\Omega_0(t_i)$ in $\{t=t_i\}$ diverges as $t_i$ tends to zero, and the inclusion $\Omega_0(t_i)\subset \Omega_i(t_i)$ proves the claim.

	\end{remark}
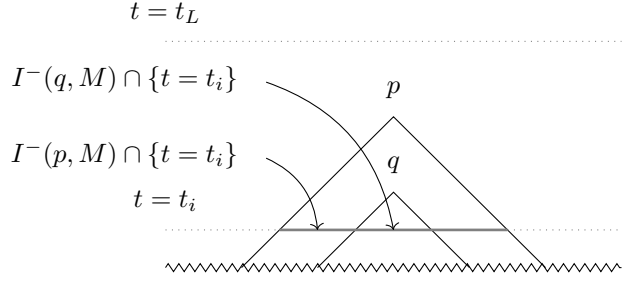
\begin{figure}
	\centering
\begin{tikzpicture}
\draw[decorate,decoration={zigzag, amplitude=1.5, segment length=4, post length=0}] (-1,0)--(5,0);
\draw (1,0)--(2,1) node[label=above:{$q$}] {} --(3,0);
\draw (0,0) -- (2,2) node[label=above:{$p$}] {} -- (4,0);
\draw [gray,text=black,dotted] (-1,0.5) node[label=above:{$t=t_i$}] {}--(0.5,0.5) (3.5,0.5)--(5,0.5);
\draw [gray,line width=1pt] (0.5, 0.5) -- (1, 0.5) node[coordinate] (outer) {} -- (2,0.5) node[coordinate] (inner) {} -- (3.5,0.5);

\draw[color=gray,text=black, dotted] (-1,3) node[label=above:{$t=t_L$}]{} -- (5,3) ;
\node[label=left:{$I^-(q,M)\cap\{t=t_i\}$}] (innerlab) at (0.2,2.5) {};
\node[label=left:{$I^-(p,M)\cap \{t=t_i\}$}] (outerlab) at (0.2,1.5) {};
\path[->] (outerlab) edge [bend left] (outer);
\path[->] (innerlab) edge [bend left] (inner);
\end{tikzpicture}
\caption{An illustration of the definition of an \emph{expanding singularity}.}
\label{fig:expanding}
\end{figure}

\section{Nonexistence of spacelike boundaries}
 \label{s18VI17.11}

\medskip

We have the following result, which immediately implies Theorem~\ref{T12VI17.1-}:

\begin{theorem}
  \label{T12VI17.1}
Suppose that $(\mcM,g)$ contains an expanding singularity towards the past.
 Then in every continuous past
  extension of $\mcM$ the boundary $\partial \mcM$ is null wherever differentiable.
\end{theorem}

The following proposition gives some more information about extensions in the current context:

\begin{proposition}
	\label{p24VIII17.1}
	Under the hypotheses of Theorem~\ref{T12VI17.1}, every spacelike hypersurface in the extension which intersects $∂^-ι(M)$ also intersects $ι(\{t<t_L\})\subseteq ι(M)$.
\end{proposition}

For a general extension (\ie where $∂^+ι(M)$ might be non-empty) we have:

\begin{theorem}\label{T16VI17.1}
Suppose that $(\mcM,g)$ contains  an expanding singularity towards the past.
If the Cauchy hypersurface $\hyp$ is compact then the past boundary $∂^-ι(M)$ of $\mcM$ in every extension is null  wherever differentiable.
\end{theorem}

\begin{proposition}
 \label{p18VI17.2}
Under the hypotheses of Theorem~\ref{T16VI17.1}, every spacelike hypersurface in the extension which intersects $∂^-ι(M)$ also intersects $ι(\{t<t_L\})\subseteq ι(M)$.
\end{proposition}

Under weaker conditions on $\hyp$ we obtain a similar but more involved result:

\begin{theorem}
 \label{T27VI17.1}
Suppose that $(\mcM,g)$ contains  an expanding singularity towards the past. If either
\begin{itemize}
\item there exist constants $t_C>0, C>0$ such that for all timelike curves $γ$ in $M$ the intersection $I^+(γ,M)\cap \{t=t_C\}$ has diameter less than $C$,
\item or for all timelike curves $γ\subset M$ the intersection
\[
W:=I^+(I^-(I^+(\gamma,I^-(\hyp, M)),M),M) \cap \hyp
\]
(compare Figure~\ref{fig:W}) is precompact.
\end{itemize}
Then:
\begin{itemize}
\item every spacelike hypersurface in the extension which intersects $∂^-ι(M)$ also intersects $ι(\{t<t_L\})\subseteq ι(M)$
\item and for every point $p∈ ∂^-ι(M)$ there exists a neighborhood $\tilde{O}\subseteq \tilde{M}$ of $p$ such that the hypersurface
    \begin{equation}\label{20X17.1}
     \{q∈∂^-ι(M)\cap \tilde{O} \,|\, I^+(q,\tilde{O})\cap ∂^-ι(M)=\emptyset\}
     \,,
    \end{equation}
     which will be referred to as the \emph{futuremost part} of $∂^-ι(M)$ in $\tilde{O}$, is null wherever it is differentiable.
\end{itemize}
\end{theorem}

\begin{figure}
	\centering
	\begin{tikzpicture}
	\definecolor{fillcolor}{rgb}{0.9,0.9,0.9}
	\draw[decorate,decoration={zigzag, amplitude=1.5, segment length=4, post length=0}] (-6,0)--(6,0);
	\draw [black, dashed] plot [smooth] coordinates {(0,2) (0.1,1.2) (-0.1,0.5) (0,0)};
	\draw (0,2) node[label=above:{$\gamma$}] {} ;
	\draw[color=gray,text=black, dotted] (-6,1.7) node[label=above:{$\hyp$}]{} -- (6,1.7) ;
	
	\draw[color=gray] (1.7,1.7) -- (0,0) --(-1.7, 1.7);
	\draw[color=gray] (3.4,0)--(1.7,1.7)  (-3.4,0) --(-1.7, 1.7);
	\draw[color=gray] (5.1,1.7) -- (3.4,0)  (-3.4,0) --(-5.1, 1.7);
	
	\node[label=right:{$W$}] (Wlab) at (-2,2.5) {};
	\path[->] (Wlab) edge [bend right] (-3.5,1.75);
	
	\draw[color=black,line width=1.0pt] (-5.1,1.7) -- (5.1,1.7);
	
	\end{tikzpicture}
	\caption{The set $W$ appearing in the second condition of Theorem \ref{T27VI17.1}.}
	\label{fig:W}
\end{figure}
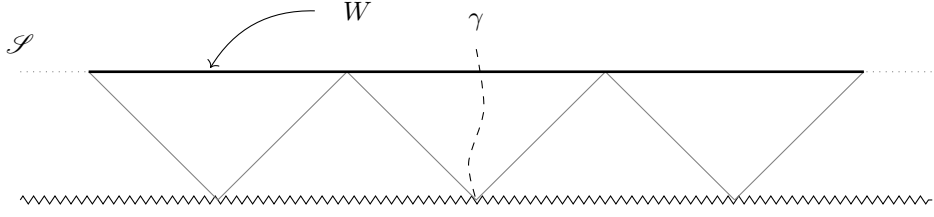

\begin{remark}
The first condition in Theorem \ref{T27VI17.1} implies the second one.
\end{remark}

\begin{figure}
	
	\begin{subfigure}{.49 \textwidth}
		\centering
		\begin{tikzpicture}
		\definecolor{fillcolor}{rgb}{0.9,0.9,0.9}
		\fill[fillcolor]  (0,0) -- (1,1) -- (0,2) -- (-1,1) -- cycle;
		\draw[-{Latex[length=2.5mm]},color=red] (0,0)--(-1,1);
		\draw[-{Latex[length=2.5mm]},color=blue] (-1,1)--(0,2);
		\draw[-{Latex[length=2.5mm]},color=red] (1,1)--(0,2);
		\draw[-{Latex[length=2.5mm]},color=blue] (0,0)--(1,1);

		\draw[color=gray,dashed] (-0.3,-0.3)--(0.3,-0.3)--(0.3,0.3)--(-0.3,0.3)--cycle;
		\node[label=right:{$\tilde{O}$}] (label) at (1,0) {};
		\path[->] (label) edge [bend left] (0.1,-0.1);
		
		\node at (0,1) {$ι(M)$};
		\end{tikzpicture}
	\end{subfigure}
	\begin{subfigure}{.49 \textwidth}
		\centering
		\begin{tikzpicture}[dot/.style={draw,circle,minimum size=1.5mm,inner sep=0pt,outer sep=0pt,fill=black}]
		\draw[color=gray,dashed] (-1,-1)--(-1,1)--(1,1)--(1,-1)--cycle;
		\draw[color=gray] (-1,-1)--(0,0)--(1,-1);
		\draw[line width=1.0pt] (-1,1)--(0,0)--(1,1);
		\node[dot, label=above:{$p$}] (0,0) {};
		
		\node[anchor=west,text width=3cm] (label) at (1.3,0.7) {futuremost part of $∂ι(M)\cap\tilde{O}$};
		\path[->] (label) edge [bend left] (0.4,0.4);
		\node at (1.3,-0.8) {$\tilde{O}$};
		\end{tikzpicture}
\end{subfigure}
	\caption{An example of a case where the futuremost part of $∂^-ι(M)$ in $\tilde{O}$ is not the same as $∂^-ι(M)\cap \tilde{O}$. The picture on the left shows the whole spacetime, with the extension given by gluing the edges according to the arrows. The picture on the right shows a neighborhood $\tilde{O}$ of the point $p\in∂^-(M)$. Here $∂^-ι(M)=∂^+ι(M)=∂ι(M)$ and there is no neighborhood $\tilde{O}$ of $p$ such that $∂ι(M)\cap\tilde{O}$ is achronal. The futuremost part, however, is achronal.}
	\label{fig:futuremost}
\end{figure}
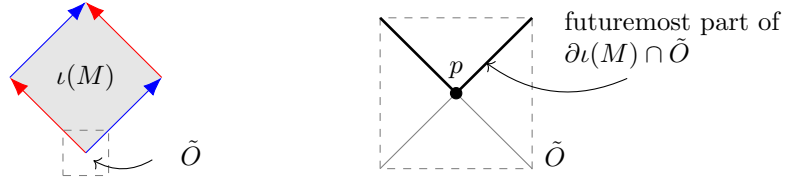

Figure~\ref{fig:futuremost} illustrates the difference between the \emph{futuremost part} of $∂^-ι(M)$, which appears in the second consequence of Theorem \ref{T27VI17.1}, and $∂^-ι(M)$ itself.

\section{Achronality of boundaries}
 \label{s2-X17.11}

\medskip

In order to prove the results above, we start by showing that the relevant sets are achronal. As such, achronal sets are locally Lipschitz continuous (for every point we can find a coordinate neighborhood on which the metric is close to Minkowski, which provides a Lipschitz bound) and therefore differentiable almost everywhere by Rademacher's theorem. By achronality they are null or spacelike wherever differentiable. To finish the proofs we will need to rule out spacelike tangent planes; this will be done in the next section.

We have:

\begin{lemma}\label{lem:one_dir}
The boundary of a future or past extension is an achronal topological hypersurface.
\end{lemma}

\begin{proof}
We assume a past extension. For a future extension, replace $I^\pm$ with $I^\mp$ in the following.

Achronality of $∂ι(M)$ is equivalent to $I^+(∂ι(M),\tilde{M})\cap ∂ι(M)=\emptyset$.  We first show that
$$
 I^+(∂ι(M),\tilde{M})\subseteq I^+(ι(M),\tilde{M})
  \,.
$$
Indeed, for every point $p\in I^+(∂ι(M),\tilde{M})$ there is $q\in ∂ι(M)$ such that $p \in I^+(q,\tilde{M})$. Therefore $q\in I^-(p,\tilde{M})$ and, because $I^-$ is open and $q$ lies on the boundary of $ι(M)$, there is a point $ι(M)\ni r\in I^-(p,\tilde{M})$. This implies $p\in I^+(r,\tilde{M})$, which proves the claim.

As $\tilde{M}$ is a past extension, no future directed timelike curve leaves $M$, \ie $I^+(r,\tilde{M})\subseteq ι(M)\,\, ∀ r\in ι(M)$, and therefore
$$
 I^+(∂ι(M),\tilde{M})\subseteq ι(M)\subseteq \tilde{M}\setminus ∂ι(M)
$$
which implies that $∂ι(M)$ is achronal.

To show that it is a topological hypersurface we need to show $∂ι(M)\cap \operatorname{edge}( ∂ι(M))=\emptyset$ (see \cite[Section 14, Lemma 25]{oneill}, and note that the proof there does not use differentiability of the metric). We consider a point $p∈ ∂ι(M)$. By the above, $I^+(p,\tilde{M})\subseteq ι(M)$ and, as $∂^+(ι(M))=\emptyset$, $I^-(p,\tilde{M})\subset \tilde{M}\setminus ι(M)$. Therefore every (past directed) timelike curve from $I^+(p,\tilde{M})$ to $I^-(p,\tilde{M})$ has to cross $∂^-(M)\subseteq ∂ι(M)$ and so $p\notin \operatorname{edge}(∂ι(M))$.
\end{proof}

\begin{remark}
A very similar result was shown by Galloway and Ling~\cite[Theorem 2.6]{GallowayLing}.
\end{remark}

Next we show that every point on an achronal boundary is the limit of a timelike curve in $ι(M)$. For a general boundary this is only true for at least one point~\cite[Lemma 2.17]{SbierskiSchwarzschild}.

\begin{lemma}\label{lem:achronal_timelike}
Given a point $p\in∂ι(M)$ and a neighborhood $p\in \tilde{O}\subset \tilde{M}$ such that $∂ι(M)\cap \tilde{O}$ is achronal in $\tilde{O}$
 then $p$ is the end point of a differentiable timelike curve in $M$, \ie there exists a differentiable timelike curve $σ: [-1,0]\to \tilde{M}$ such that $σ([-1,0))\subset ι(M)$ and $σ(0)=p$.
\end{lemma}

\begin{proof}
The curve $σ$ might be future- or past-directed, we consider both cases here.
We choose coordinates $(x^α)\in(-δ,δ)\times(-ε,ε)^{d}$ on $\tilde{O}$ such that $p=(0,\dots,0)$, the metric is close to Minkowski and the sets $\{\pm \frac{19}{20}δ\}\times(-ε,ε)^d$ are contained in $ I^\pm(p,\tilde{O})$ (by choosing $\tilde{O}$ smaller if necessary).

For every $\underline{x}\in(-ε,ε)^d$ there is at most one $t\in(-δ,δ)$ such that $(t,\underline x)\in ∂ι(M)$: If there were more than one then they could be connected by a (vertical) timelike curve, contradicting achronality.

As $p$ is a boundary point there is a point $\tilde O\ni q\in ι(M)$. It can be connected by a vertical curve which does not cross $∂ι(M)$ to one of $\{\pm \frac{19}{20}δ\}\times(-ε,ε)^d$. Therefore the differentiable timelike curve $σ:(0,\frac{19}{20}δ)\to\tilde{O}, σ(s)=(\pm s,0,\dots,0)$ lies in $ι(M)$, as it can be connected to $q$ by a curve which does not cross the boundary.
\end{proof}

\begin{remark}
	Lemma \ref{lem:one_dir} and Lemma \ref{lem:achronal_timelike} together imply that if $\partial^+ ι(M)=\emptyset$ then $∂ι(M)=∂^- ι(M)$, with a similar statement obtained by reversing time-orientation.
\end{remark}

In order to prove Theorems \ref{T16VI17.1} and \ref{T27VI17.1} we will need the following two Lemmas.

\begin{lemma}\label{lem:horizon_achronal}
	We consider a globally hyperbolic spacetime $(M, g)$ and a $C^0$-extension $(\tilde{M}, \tilde{g})$ with embedding $ι:M\to \tilde{M}$. If there exists a spacelike Cauchy hypersurface $\hyp\subset M$ such that for all timelike curves $\gamma\subset M$ the set
	\[
	W:=I^+(I^-(I^+(\gamma,I^-(\hyp, M)),M),M)\cap \hyp
	\]
	is precompact in $M$, then for any $p∈∂^-ι(M)$ there exists a neighborhood $\tilde{O}$ of $p$ such that the futuremost part of $∂^-ι(M)$ in $\tilde{O}$, $\{q∈∂^-ι(M)\cap \tilde{O} \,|\, I^+(q,\tilde{O})\cap ∂^-ι(M)=\emptyset\}$, is a non-empty achronal topological hypersurface.
 \end{lemma}
Recall that the set $W$ is shown in Figure~\ref{fig:W}.
\begin{remark}
	A simple condition that gives a precompact $W$ is an upper bound for the diameter of $I^+(\gamma, M)\cap \hyp$ which is uniform with respect to $\gamma$.
\end{remark}
\begin{proof}
	Let $p∈∂^-ι(M)$ be a point on the past boundary, $γ$ the timelike curve approaching it, and $\hyp$ the spacelike Cauchy hypersurface as in the Lemma. As $\bar{W}$ (closure in $M$) is compact we can choose a neighborhood $\tilde{O}$ of $p$ such that $\tilde{O}\cap ι(\bar{W})=\emptyset$. Choosing a smaller neighborhood if necessary we introduce coordinates  $(x^α)\in(-δ,δ)\times(-ε,ε)^{d}$ such that the metric is close to Minkowski, $\dot{γ}=∂_{x^0}$ and $\{\pm \frac{19}{20}δ\}\times(-ε,ε)^d\subseteq I^\pm(p,\tilde{O})$.

We consider the straight (in coordinates) timelike curves $σ_{\underline{x}}: (-δ,\frac{19}{20}δ]\to \tilde{O},\; s\mapsto (s, \underline{x})$ where $\underline{x}\in (-ε,ε)^d$. These end in $ι(M)$ by the construction of the coordinates. We define
$$
 b(\underline{x})=\inf \{s\in(-\delta,\frac{19}{20}δ]\;|\; ∀ s'>s, σ_{\underline{x}}(s')\subset ι(M)\}
$$
and show that $b(\underline{x})>-δ$, \ie that the $σ_{\underline{x}}$ intersect $∂^-ι(M)$ at least once, at $σ_{\underline{x}}(b(\underline{x}))$, and that
$$
 σ_{\underline{x}}(b(\underline{x}))∈\tilde O \setminus (I^+(p, \tilde{O})\cup I^-(p,\tilde{O})
 )
  \,.
$$
	
	As $\hyp$ is a Cauchy hypersurface and by the definition of $W$, the future in $\tilde{O}$ of any point in $I_γ\cap \tilde{O}$, where
$$
 I_γ:=ι(I^-(I^+(γ,M)\cap I^-(\hyp,M),M))\,,
$$
has to be contained in $ι(M)$.
In particular, as the future of a point is the future of any past directed curve ending at that point and $γ\cap\tilde{O}\subset I_γ$, $I^+(p,\tilde{O})=I^+(γ,\tilde{O})\subset ι(M)$ and therefore
$$
 σ_{\underline{x}}(b(\underline{x}))\notin I^+(p,\tilde{O})
  \,.
$$
%

We now have to show that $b(\underline{x})>-\delta$ and $σ_{\underline{x}}(b(\underline{x}))\notin I^-(p,\tilde{O})$. If we assume this is false then there exists $s^-\in(b(\underline{x}),\frac{19}{20}δ)$ such that $σ_{\underline{x}}(s^-)\in I^-(p,\tilde{O})$. But we also have $σ_{\underline{x}}(s^-)\subset I_γ$: The end point $σ_{\underline{x}}(\frac{19}{20}δ)$ of $σ_{\underline{x}}$ is contained in $\{\frac{19}{20}δ\}\times(-ε,ε)^d\subseteq I^+(p,\tilde{O})=I^+(γ,\tilde{O})$, and therefore $σ_{\underline{x}}((b(\underline{x}),\frac{19}{20}δ])\subset I_γ$. Now, as $σ_{\underline{x}}(s^-)\subset I_γ$ and $p\in I^+(σ_{\underline{x}}(s^-),\tilde{O})$, we have $p\in ι(M)$ by the argument in the previous paragraph, which is a contradiction.

	Repeating the argument with $p$ replaced by $σ_{\underline{x}}(b(\underline{x}))$ for all $\underline{x}∈(-ε,ε)^d$ shows that the set $\{σ_{\underline{x}}(b(\underline{x})) \;|\; \underline{x}∈(-ε,ε)^d\}=\{q∈∂^-ι(M) \;|\; I^+(q,\tilde{O})\cap ∂^-ι(M)=\emptyset\}$ is achronal in $\tilde{O}$. It is a topological hypersurface as $p$ can't be an edge point by the properties of $b(\underline{x})$.
\end{proof}

Under the stronger assumption that the Cauchy hypersurface $\hyp$ is compact, we obtain the following simpler result.

\begin{lemma}\label{lem:compact_past_achronal}
	We consider a globally hyperbolic spacetime $(M, g)$ and a $C^0$-extension $(\tilde{M}, \tilde{g})$ with embedding $ι:M\to \tilde{M}$. If there exists a compact spacelike Cauchy hypersurface of $M$ then $∂^-ι(M)$ is a locally achronal topological hypersurface.
\end{lemma}

\begin{proof}
	Let $\hyp$ be the compact spacelike Cauchy hypersurface. As any subset of a compact set is precompact by definition, Lemma \ref{lem:horizon_achronal} applies. What remains to be shown is that the set $\{q∈∂^-ι(M)\cap \tilde{O} \,|\, I^+(q,\tilde{O})\cap ∂^-ι(M)=\emptyset\}$ is actually the full $∂^-ι(M)\cap \tilde{O}$, \ie that there are no additional points in $∂^-ι(M)$ below it.
	
	We choose a neighborhood $\tilde{O}$ and coordinates as in the proof of Lemma \ref{lem:horizon_achronal}, but with $\tilde{O}\cap \hyp = \emptyset$. Now the future in $\tilde{O}$ of any point $q\in ∂^-ι(M)\cap \tilde{O}$ lies in $ι(M)$.
	
	If there was a point $q\in ∂^-ι(M)\cap \tilde{O}$ such that $∂^-ι(M) \ni p\in I^+(q,\tilde{O})\cap ∂^-ι(M)$ then $p$ would lie in $ι(M)$.
\end{proof}

\section{Proofs of the main theorems}\label{Sproofs}

\medskip

The proofs depend on the following Lemma, which is a slight variation of a result of Sbierski~\cite{SbierskiAlternate}. We use his notation for the sets
$$
C_a^-:=\{0\ne X\in \R^{d+1}| \frac{<X,e_0>_{\R^{d+1}}}{|X|_{\R^{d+1}}}<-a\}
\,,
$$
where $0<a<1$, $<.,.>_{\R^{d+1}}$ is the Euclidean scalar product in $\R^{d+1}$ and $|.|_{\R^{d+1}}$ is the Euclidean norm. The $C_a^-$ are cones of vectors with angle less than $\cos^{-1}(a)$ to the $x^0$ axis with the tip of the cone pointing up.

\begin{lemma}\label{lem:distbound}
	We consider a spacetime $(M,g)$ with extension $(\tilde M, \tilde g)$. Given a neighborhood $\tilde{O}$ of a point on the boundary of $ι(M)$, a point $p\in \tilde{O}\cap ι(M)$ and a chart $\tilde\psi: \tilde O \to (-δ,δ)\times (-ε,ε)^d$ such that

	\begin{enumerate}
		\item $∂_{x_0}$ is timelike,
		\item $|\tilde{g}_{αβ}- η_{αβ}|<ν$ where $η$ is the Minkowski metric and $1/2>ν>0$ a constant such that $∀ a\in (-δ,δ)\times (-ε,ε)^d,\; \tilde{ψ}^{-1}(a+C^-_{5/6})\subseteq I^-(\tilde{ψ}^{-1}(a),\tilde{O})\subseteq \tilde{ψ}^{-1}(a+C^-_{5/8})$,
		\item $\tilde{ψ}^{-1}(\{\tilde{x}^0<-\frac{1}{10}δ\})\subseteq \tilde{M}\setminus I^-(\tilde{ψ}^{-1}(\{\frac{19}{20}δ\}\times (-ε,ε)^d),ι(M))$,
		\item $I^-(p,ι(M))\subset \tilde{O}$,
		\item $\tilde{\psi}^{-1}(\{\frac{19}{20}δ\}\times (-ε,ε)^d)\subseteq ι(I^+(p, \tilde{O}))$,
	\end{enumerate}
there exist $q ∈ I^-(p,ι(M))\cap\tilde{O}$ and a constant $0<C_d<\infty$ such that for all Cauchy hypersurfaces $\hyp$ of $M$ the distance in $I^-(p,ι(M))\cap ι(\hyp)$ of any two points in $I^-(q,ι(M))\cap ι(\hyp)$ is bounded above by $C_d$.
\end{lemma}
\begin{proof}
	We choose $q\in I^-(p,\tilde{O})\cap ι(M)$ such that
	\[
	(\tilde{ψ}(q)+C^-_{5/8})\cap\{\tilde{x}^0>-δ/10 \} \subset (\tilde{ψ}(p)+C^-_{5/6})\cap\{\tilde{x}^0>-δ/10\}\,,
	\]
	\ie such that the past of $q$ in $M$ lies completely inside a (Euclidean) cone contained in the past of $p$.
	
	We assume  $q\in ι(I^+(\hyp, M))$, as otherwise $I^-(ι^{-1}(q),M)\cap \hyp = \emptyset$ and there is nothing to show. By property (iii) of the chart $\tilde{ψ}$ there exists a function $(-ε,ε)^d\to (-δ,δ)\,,\; \underline{x}\mapsto L_{\underline{x}}$ such that the timelike curves $σ_{\underline{x}}: (L_{\underline{x}}, δ)\to M, σ_{\underline{x}}(s)=ι^{-1}(\tilde{ψ}^{-1}(s,\underline{x}))$ are past inextendible in $M$. As $q$ lies in the past of $p$, property (v) holds for $q$ as well. This implies that the curves $σ_{\underline{x}}$ intersect $I^+(\hyp,M)$ and therefore intersect $\hyp$ exactly once, say at $s=f(\underline{x})$. We take this as the definition of $f: (-ε,ε)^d \to (-δ,δ)$.
	
	We first show that $f$ is smooth:
	As $\tilde{ψ}(ι(\hyp)\cap \tilde{O})$ is a smooth submanifold, there exists for every point $(f(\underline{x}_0),\underline{x}_0)$ a neighborhood $W$ and a smooth submersion $g: W\to \R$ such that $\tilde{ψ}(ι(\hyp)\cap \tilde{O})\cap W=\{g=0\}$. As $\hyp$ is a Cauchy hypersurface no timelike vector can be tangent to it. Therefore $∂_0 g|_{(f(\underline{x}_0), \underline{x}_0)}\neq 0$ and by the implicit function theorem there exists a smooth function $h: (-ε, ε)^d \supseteq V\to (-δ,δ)$, where $V$ is a neighborhood of $\underline{x}_0$, such that $g(h(\underline{x}), \underline{x})=0$. Thus $f|_V=h$ and therefore $f$ is smooth.
	
	The next step is to show that $|∂_i f|$ is bounded by a positive constant in $\tilde{O}$ for all $i$. As vectors tangent to $\hyp$ cannot be timelike we obtain the inequality
	\begin{equation}\label{2017VI6.1}
	0\leq \tilde{g}((∂_i f)∂_0+∂_i, (∂_i f) ∂_0+∂_i)=(∂_i f)^2 \tilde{g}_{00}+2(∂_i f)\tilde{g}_{0i}+\tilde{g}_{ii}\,.
	\end{equation}
	By property (ii), $\tilde{g}_{00}<-1/2$, $\tilde{g}_{ii}>1/2$, and $\tilde{g}_{0i}<1/2$ and therefore this inequality is only satisfied for $(∂_i f)_-\leq (∂_i f)\leq (∂_if)_+$ where $(∂_i f)_\pm$ are the values where equality holds in \eqref{2017VI6.1}, \ie
	\[
	(∂_i f)_\pm = \frac{-\tilde{g}_{0i}\mp\sqrt{(\tilde{g}_{0i})^2-\tilde{g}_{ii}\tilde{g}_{00}}}{\tilde{g}_{00}}\,.
	\]
	Again using property (ii), we see that $(∂_i f)_\pm$, and therefore also $(∂_i f)$, are bounded by a constant independent of $f$.
	
	We define $ω: (-ε, ε)^d \to (-δ,δ) \times (-ε, ε)^d,\; ω(\underline{x})=(f(\underline{x}), \underline{x})$. This parameterizes a smooth submanifold $\tilde{S}$ which is isometric to an open subset of $\hyp$ in $M$ by $ι^{-1}\circ \tilde{ψ}^{-1}$. We denote by $\bar{g}$ the metric induced on $\tilde{S}$ by $\tilde{g}$. The components of $\bar{g}$ are
	\[
	\bar{g}_{ij}=\tilde{g}_{μν}\frac{∂ω^μ}{∂x_i}\frac{∂ω^ν}{∂x_j}=\tilde{g}_{00}\frac{∂f}{∂x_i}\frac{∂f}{∂x_j}+\tilde{g}_{0j}\frac{∂f}{∂x_i}+\tilde{g}_{i0}\frac{∂f}{∂x_j}+\tilde{g}_{ij}\,.
	\]
	As $|\tilde{g}_{μν}|$ and $∂_i f$ are bounded by the above, we have $|\bar{g}_{ij}|<C_{\bar{g}}$ for a positive constant $C_{\bar{g}}$.
	
	We now consider two points $r,s\in I^-(q, ι(M))\cap ι(M)$ as in the lemma. As $I^-(q,M)\subset \tilde{O}$ there exist $\underline{x}, \underline{y}∈ (-ε, ε)^d$ such that $ω(\underline{x})=\tilde{ψ}(ι(r))$ and $ω(\underline{y})=\tilde{ψ}(ι(s))$. The length of the straight line $σ: [0,1] \to (-ε,ε)^d$, $σ(\ell)=\underline{x}+ \ell(\underline{y} - \underline{x})$ is given by
	\[\begin{split}
	L(σ)=&\int_0^1 \sqrt{\bar{g}(\dot{σ}(\ell), \dot{σ}(\ell))}\;\D\ell\\
	=&\int_0^1\sqrt{\sum_{i,j=1}^d(\underline{y}_i-\underline{x}_i)\bar{g}_{ij}(σ(\ell))(\underline{y}_j - \underline{x}_j)}\;\D\ell\\
	\leq& \int_0^1\sqrt{\sum_{i,j=1}^d 2ε \cdot C_{\bar{g}}\cdot 2ε}\;\D\ell\\
	=&2ε d \sqrt{C_{\bar{g}}}\,.
	\end{split}\]
	
	As $I^-(q,ι(M))\cap \tilde{O}\subset \tilde{ψ}^{-1}(\tilde{ψ}(p)+C^-_{5/6})$ the curve $σ$ is contained in $I^-(p,M)\cap\hyp$. As the distance in $I^-(p,M)\cap \hyp$ between $r$ and $s$ is defined as the infimum over the length of all piecewise smooth curves connecting them, and $ι^{-1}\circ\tilde{ψ}^{-1}\circ σ$ is one such curve, this completes the proof.
\end{proof}

We are now ready to prove Theorem \ref{T12VI17.1}.

\begin{proof}[Proof of Theorem \ref{T12VI17.1}]
By Lemma \ref{lem:one_dir} the boundary $∂ι(M)$ of the extension is achronal, hence differentiable almost everywhere, and by Lemma \ref{lem:achronal_timelike} we have for every $p_∂∈∂ι(M)$ a timelike curve $γ: [0,1]\to \tilde{M}$ such that $γ([0,1))\subset ι(M)$ and $γ(1)=p_∂$.

We assume that $∂ι(M)$ is spacelike at a point $p_∂$ at which it is differentiable and establish a contradiction. We choose a neighborhood $\tilde{O}$ of $p_∂$ such that $t<t_L$
 in $\tilde{O}\cap ι(M)$, where $t_L$ is the constant appearing in the Definition~\ref{def:expanding} of expanding singularity. This is possible as $∂ι(M)$ is achronal: In any neighborhood of $p_∂$, either $t<t_L$ or $\{t=t_L\}$ is some achronal set lying above $∂ι(M)$ and we can find a smaller neighborhood of $p_∂$ which doesn't intersect it.

We choose coordinates $\tilde{x}^i$ on $\tilde{O}$ such that $\tilde{x}^i(p_∂)=0$, $\tilde{g}(p_∂)=η$ (where $η$ is the Minkowski metric) and $|\tilde{g}(x)_{αβ}-η_{αβ}|<ν$ for all $x\in\tilde{O}$ and $ν$ such that condition 2 of Lemma \ref{lem:distbound} is satisfied.

We perform a Lorentz boost to transform the normal vector of $∂ι(M)$ at $p_∂$ to $∂_{x^0}$ and, by choosing a smaller neighborhood $\tilde{O}$ if necessary, ensure that $∂ι(M)$ is almost horizontal, \ie $∂ι(M)\cap \tilde{O}\subset \{-δ/10<\tilde{x}^0<+δ/10\}$ and therefore condition 3 is satisfied.

We choose a point $p \in \tilde{O}\cap ι(M)$ such that $I^-(p,ι(M))\subset\tilde{O}$, satisfying condition 4. Finally, by choosing a smaller neighborhood, we can satisfy the remaining condition 5.

Applying Lemma \ref{lem:distbound} we obtain, for a point $q\in I^-(p,ι(M))$ and for all Cauchy hypersurfaces $\hyp$, an upper bound for the distance in $I^-(p,ι(M))\cap\hyp$ of any two points $r,s \in I^-(q,M)\cap \hyp$.

This is a contradiction to Definition \ref{def:expanding} of expanding singularities.
\end{proof}
\begin{proof}[Proof of Proposition \ref{p24VIII17.1}]
To prove that there exists no spacelike hypersurface $Σ_{p_∂}\subset \tilde{M}\setminus ι(M)$ such that  $p_∂\in Σ_{p_∂}$ we use the same argument, but choose coordinates in $\tilde{O}$ such that $Σ_{p_∂}$ instead of $∂ι(M)$ is almost horizontal. As $Σ_{p_∂}$ is spacelike, the part of $I^-(p,\tilde{M})$ lying above $Σ_{p_∂}$ is entirely contained in $\tilde{O}$, we denote it by $I_{Σ_{p_∂},p}$. As any past directed curve from $p$ has to cross $Σ_{p_∂}$ before leaving $\tilde{O}$ we have $I^-(p,ι(M))\subseteq I_{Σ_{p_∂},p}$ and therefore we obtain a contradiction as before.
\end{proof}

The proofs of Theorems \ref{T16VI17.1} and \ref{T27VI17.1} proceed  in a very similar way:

\begin{proof}[Proof of Theorems \ref{T16VI17.1} and \ref{T27VI17.1}]
By Lemma \ref{lem:compact_past_achronal} the conditions of Theorem \ref{T16VI17.1} imply that $∂^-ι(M)$ is a locally achronal topological hypersurface. Similarly, under the conditions of Theorem \ref{T27VI17.1}, Lemma \ref{lem:horizon_achronal} implies that the ``locally futuremost part'' of the past boundary $∂^-ι(M)$ is  achronal, \ie for every $p\in ∂^-ι(M)$ there exists a neighborhood $U$ of $p$ such that
$$
 \{z∈ ∂^-ι(M)\cap U\; |\; I^+(z,U)\cap ∂^-ι(M) = \emptyset\}
$$
is an achronal topological hypersurface in $U$.

By the definition of $∂^-ι(M)$ we have for every $p∈ ∂^-ι(M)$ a timelike curve $γ$ ending at $p$. The proof now proceeds analogously to that of Theorem \ref{T12VI17.1} with $∂ι(M)$ replaced by $∂^-ι(M)$ in the case of Theorem \ref{T16VI17.1} and with the ``locally futuremost part'' of $∂^-ι(M)$ in the case of Theorem \ref{T27VI17.1}.

The argument excluding the existence of a spacelike hypersurface $Σ_{p_∂}\subset \tilde{M}\setminus ι(M)$ such that  $p_∂\in Σ_{p_∂}$ follows the proof of Proposition \ref{p24VIII17.1}, except that to guarantee that $I^-(p,ι(M))$ lies above $Σ_{p_∂}$ we only need $Σ_{p_∂}\cap ι(\{t<t_L\})=\emptyset$ as $t<t_L$ in $I^+(∂^-ι(M),\tilde{O})$.
\end{proof}

We pass now to the

\begin{proof}[Proof of Proposition~\ref{P30X17}]
	In contrast to the proofs above,  if the boundary is null we cannot use the pasts of points in $M$ to identify subsets of Cauchy hypersurfaces in neighborhoods in $\tilde{M}$ with those in $M$.
Instead we consider the total diameter of the (compact) Cauchy hypersurface. We first need to ensure that neighborhoods of the boundary contain the whole (embedded) Cauchy hypersurface $\hyp_s:=\{t=s\}$, for $s$ sufficiently small.
	
	By the proof of Theorem 3.1 in \cite{chrusciel_1993} there exists a constant $δ>0$ such that for all $0<s\leq δ$ there is a map $ψ_s: \tilde{M}\to \tilde{M}$, the flow of a continuous timelike vector field on $\tilde{M}$, such that $ψ_s(∂ι(M))=ι(\hyp_s)$ and these are compact. In \cite{chrusciel_1993} it is assumed that the extended manifold is at least $C^3$, but this can be relaxed in our case: The results of \cite[Lemma 3.2]{chrusciel_1993} hold by our  Lemma~\ref{lem:one_dir}. Note that \cite[Lemma 2.1]{chrusciel_1993} and \cite[Lemma 3.3]{chrusciel_1993} require only differentiabilty of the manifold, not the metric. The rest of the proof takes place in $(M,g)$ which is smooth in any case. In addition, \cite[Theorem 3.1]{chrusciel_1993} assumes that the boundary itself is $C^1$ but this can similarly be relaxed as long as the Cauchy hypersurfaces $\hyp_t$ are $C^1$.
	
	By Lemma~\ref{lem:one_dir} and Lemma~\ref{lem:achronal_timelike} the boundary $∂ι(M)=∂^-ι(M)$ of the extension is achronal. We choose for each point $p\in ∂ι(M)$ an open neighborhood in $\tilde{M}$ and coordinates such that the metric is close to Minkowski. As $∂ι(M)$ is compact we can find a finite subcover $\{O_i\}$. For each $i$ there exists a constant $\delta_i$ such that for each $p\in \bar{O}_i\cap ∂ι(M)$ and all $0<s\leq \delta_i$,  $ψ_s(p)\in \bar{O}_i$. As there are only finitely many $i$ we can set $δ_{min}=\min \{δ_i\}$ and obtain $\hyp_s=ψ_s(∂ι(M))\subset \bigcup_i \bar{O}_i$ for all $0<s\leq \delta_{min}$.
	
	The distance in $\hyp$ between any two points of a Cauchy hypersurface $\hyp$ is bounded above in each $O_i$ by a constant independent of $\hyp$ by a similar argument as in the proof of Lemma \ref{lem:distbound}. As $I^+(∂ι(M),\tilde{M})\subset ι(M)$ and the Cauchy hypersurfaces $\hyp_t$ are achronal in $M$ they can intersect $O_i$ only once (\ie $ι(\hyp_t)\cap O_i$ is connected), and therefore the total diameter of $\hyp_t$ is bounded. This contradicts the Definition \ref{def:expanding_global} of a globally expanding singularity.
\end{proof}

\section{Examples}
 \label{s21X16.1}

\medskip

 We will use the following Lemma to show that the examples below contain an expanding singularity towards the past, as required by Theorem \ref{T12VI17.1}:

 \begin{lemma}\label{lem:expanding_singularity}
	We consider a globally hyperbolic spacetime $(M, g)$ of dimension $n+1$ with a Cauchy time function $t: M\to (0,\infty)$ such that $M=(0,\infty)\times \hyp$. Suppose that there exists $t_L>0$ such that the subset $\{t<t_L\}\subset M$ can be covered by charts of the form
	$(0,t_L)\times U$, for some open subset $U\subseteq \hyp$, in which the metric takes the form
	\[
	g=g_{00}(t,x)\dt^2+
	g_{ij}(t,x)\dx^i\dx^j\,,
	\quad g_{00}<0
	\,,
	\]
	and satisfies
	\begin{gather}\label{6.1_cond12}
	g_{11}(t,x)\xrightarrow{t\to 0} \infty\,,\quad g_{ij}(t,x)\xrightarrow{t\to 0} 0\quad \text{for} \quad (i,j)\neq(1,1)\,,\\
	\label{6.1_cond3}
	\text{and either}\quad g_{1i}=0 \quad\text{for}\quad i \ne 1\quad\text{or}\quad \frac{g_{mi}g_{mj}}{g_{mm}g_{ij}}\xrightarrow{t\to 0} 0 \quad \text{for}\quad i,j<m
	\end{gather}
	uniformly on compact subsets in $x$. Assume moreover that either
	\begin{itemize}
		\item $\hyp$ is compact and $\{t<t_L\}$ is covered by a single chart as above,
		\item or for every $p ∈ \{t<t_L\}\subset M$ there exists a chart as above which contains $I^-(p,M)$ and a compact set $K_p\subset \hyp$ such that $I^-(p, M)\subset (0,t_σ)\times K_p$.
	\end{itemize}
	Then $(M,g)$ contains an expanding singularity towards the past.
\end{lemma}

\begin{proof}
	We need to show that for every $p \in \{t<t_L\}$ and $q\in I^-(p,M)$ there exist a sequence $t_i$ decreasing to zero such that the diameter of $I^-(q,M)\cap \{t=t_i\}$ in $I^-(p,M)\cap \{t=t_i\}$ tends to infinity as $t_i\to 0$.

	If the second condition in \eqref{6.1_cond3} holds, we start by defining an orthonormal frame $\{θ^m\}$, $g_{ij}dx^i dx^j=\delta_{ij}θ^i θ^j$ on each slice by
	\[
	θ^m=\sqrt{h^m_{mm}}(dx^m+\sum_{i=1}^{m-1}\frac{h^m_{mi}}{h^m_{mm}}dx^i)\,,
	\]
	where
	\[
	h^{m-1}=h^m-(θ^m)^2\quad\text{and}\quad h^n_{ij}:=g_{ij}\,.
	\]
	The tensor field $h^{m-1}$ is positive definite on the subspace spanned by $\{\partial_1,\ldots,\partial_{m-1}\}$: Indeed, if there were a vector $X=X^1\partial_1+\ldots X^{m-1}\partial_{m-1}$ such that $h^{m-1}(X,X)\leq0$ then we could choose $X^m$   so  that $θ^m(X+X^m\partial_m)=0$. Setting $0\ne Y:=X+X^m\partial_m$, this gives  $h^{m-1}(Y,Y)=h^{m-1}(X,X)\le 0$ and  $h^{m}(Y,Y)=h^{m-1}(X,X)\le 0$, giving a contradiction to the positive definiteness of $g$.
	
	Writing $h^{m-1}$ in terms of $h^m$ we obtain
	\[
	h^{m-1}=\sum_{i,j=1}^{m-1}\left(h^m_{ij}-\frac{h^m_{mi}h^m_{mj}}{h^m_{mm}}\right)dx^i dx^j\,.
	\]
	One checks that $h^{m-1}_{ij}/h^{m}_{ij}\to 1$ follows from \eqref{6.1_cond3},   which further implies  $h^{m}_{ij}/g_{ij}\to 1$ for $i,j\le m$.
	By construction we have  $θ^1=α(t,x)dx^1$, with $α(t,x)/\sqrt{g_{11}(t,x)}\to 1$, and with all convergences uniform on compact subsets.
	
	We now choose $k_0\in \N$ and $ε>0$ such that the sequences of points
	$$
	\mbox{
		$y_1^k=(1/k,x^1(q)-ε,x^2(q),\dots)$\, and\, $y_2^k=(1/k,x^1(q)+ε,x^2(q),\dots)$}
	$$
	 fulfill $y_1^{k}, y_2^{k}\in I^-(q,M)$ for all $k∈\{k_0,k_0+1,\dots\}$. The distance between $y_1^k$ and $y_2^k$, and therefore the diameter of $I^-(q,M)\cap\{t=1/k\}$, in $I^-(p,M)\cap \{t=1/k\}$ is bounded from below by
	\[\begin{split}
	d_{\{1/k\}\times K_p}(y^k_1,y^k_2)&=\inf_σ \int\limits_{-ε}^ε \sqrt{\sum_{i=1}^n (θ^i(\dot{σ}))^2}\ds\\
	&\geq\inf_σ \int\limits_{-ε}^ε \sqrt{(θ^1(\dot{σ}))^2}\ds=\inf_σ \int\limits_{-ε}^ε |\alpha(t,x)\dot{σ}^1|\ds\,,
	\end{split}\]
	where the infimum is taken over curves $σ: [-ε,ε]\to \{1/k\}\times K_p$ (with $K_p:=\hyp$ in the case where $\hyp$ is compact), such that $σ(-ε)=y_1^k$ and $σ(ε)=y_2^k$. If the first condition holds in \eqref{6.1_cond3} we obtain the same expression with $α:=g_{11}$.
	
	We have
	\[\begin{split}
	\lim_{k\to \infty}
	d_{\{1/k\}\times K_p}(y^k_1,y^k_2)
	&=\lim_{k\to\infty}\inf_σ
	\int\limits_{-ε}^{+ε}|\alpha(σ(s)) \dot{\sigma}^1(s)|\ds\\
	&>\lim_{k\to\infty}
	\left(\min_{x\in K_p}α(1/k,x)\right)
	\inf_σ\int\limits_{-ε}^{+ε}|\dot{\sigma}^1(s)|\ds\,,
	\end{split}\]
	where the minimum diverges by the uniform divergence of $α$ on $K_p$ and the integral is independent of $k$ and positive as $\int\limits_{-ε}^{+ε} \dot{σ}^1(s)\ds=2ε$.
\end{proof}

\begin{remark}\label{rem:sequence}
	Lemma \ref{lem:expanding_singularity} directly generalizes to the case where the expanding direction rotates within the $t=\text{const}$ hypersurfaces. Indeed, let $\{t_k\}$ be a sequence of times and $\{X_k\}$ a sequence of associated expanding directions (constant in $x$ unit vectors w.r.t.\ the Euclidean metric $δ_{ij}$) in each hypersurface $\hyp_k:=\{t=t_k\}$. A spatial rotation for each $t_k$ gives $X_k = ∂_{x^1}$ and the conditions on the metric are as before in these adapted coordinates.
	This is relevant for the case of Bianchi IX and BKL singularities discussed in Sections \ref{sec:bianchiix} and \ref{sec:genericbkl} below.
\end{remark}

\subsection{AVTD metrics without symmetries}\label{s20X17.1}
These solutions, constructed in~\cite{Chrusciel2015,Klinger2015}, take the form
\begin{equation}\label{10II15.1}
\ds^2=-e^{-2\sum_{a=1}^3  β^a}\D τ^2+\sum_{a=1}^3  e^{-2β^a}\Ni{^a_i}\Ni{^a_j}\dx^i\dx^j\,,
\end{equation}
with $β^a$ and $\Ni{^a_i}$, $i,a\in\{1,2,3\}$,  depending on all coordinates $τ$, $x^i$ and behaving asymptotically as
\begin{equation}\label{10II15.2}
β^a=β_\circ^a+τp_\circ^a+O(e^{-τν})\quad\text{and}\quad
\Ni{^a_i}=δ^a_i+O(e^{-2τ(p_\circ^i-p_\circ^a)})
 \,,
\end{equation}
with $ν>0$.
They are parameterized by freely prescribable analytic functions $β_\circ^2, β_\circ^3$ and $P\indices{_\circ^2_1}$ depending on all space coordinates and two analytic functions $p_\circ^2$ and $p_\circ^3$ depending on all space coordinates which are free except for the inequalities
\begin{equation}\label{p_ineq}
0<p_\circ^2<(\sqrt{2}-1)p_\circ^3\,.
\end{equation}
The function $P\indices{_\circ^2_1}$ does not appear in the asymptotic expansion \eqref{10II15.2} but influences lower order terms in the expansion of $\Ni{^a_i}$.

The remaining exponent $p_\circ^1$ is given by
\begin{equation}\label{17V17.1}
p_\circ^1=-\frac{p_\circ^2 p_\circ^3}{p_\circ^2+p_\circ^3}<0\,,
\end{equation}
\ie $x^1$ is the expanding direction.

The solutions approach a curvature singularity as $τ\to\infty$.

The construction of these spacetimes in~\cite{Klinger2015} was done in a purely local manner, regardless of the topology of the solutions. Here we will assume that the spatial topology is compact.

We replace the time coordinate $τ$ with $t=-\log τ$, giving, with the asymptotic expansion \eqref{10II15.2} inserted,
\begin{equation}\label{9II16.1}\begin{split}
\ds^2=&-t^{2σ_{p_\circ}-2}e^{-2\sigma_{β_\circ}}(1+O^ν)\dt^2\\
&+\sum_{a=1}^3 t^{2p_\circ^a}e^{-2β_\circ^a}(1+O^ν)(δ^a_i+O(t^{2(p_\circ^i-p_\circ^a)}))(δ^a_j+O(t^{2(p_\circ^j-p_\circ^a)}))\dx^i\dx^j
\end{split}\end{equation}
where $O^ν=O(t^ν)$, $\sigma_{p_\circ}=p_\circ^1+p_\circ^2+p_\circ^3$, and $\sigma_{β_\circ}=β_\circ^1+β_\circ^2+β_\circ^3$.
We see directly that $g_{11}\to \infty$ and $g_{ij}\to 0$, satisfying \eqref{6.1_cond12}. From \eqref{9II16.1} we find that the metric component $g_{ij}$ is of the same order as the faster decaying one of $g_{ii}$ and $g_{jj}$ and that $g_{ii}/g_{jj}\to 0$ for $i>j$. This implies that \eqref{6.1_cond3} is also satisfied, \ie the solutions are of the form required in Lemma \ref{lem:expanding_singularity}.

\subsection{$T^3$ Gowdy}

The $T^3$ Gowdy spacetimes have metrics of the form
\[
g=e^{(τ-λ)/2}(-e^{-2τ}\Dτ^2+\Dθ^2)+e^{-τ}(e^P\Dσ^2+2e^PQ\Dσ\Dδ+(e^PQ^2+e^{-P})\Dδ^2)
 \,,
\]
where $λ$, $P$ and $Q$ are functions of $τ$ and $θ$ and the singularity is approached as $τ\to\infty$. The $τ=\text{const}$ slices are toroidal and therefore compact. Ringström showed~\cite{RingstroemSCC} that generic Gowdy spacetimes asymptotically behave as follows:
\[
P(τ,θ)=v_a(θ)τ+ϕ(θ)+o(1),\qquad Q(t,θ)=q(θ)+o(1),\qquad λ=v_a(θ)^2τ+o(τ)
\]
where $0<v_a(θ)<1$ and the lower order terms converge uniformly. Therefore the $g_{θθ}$ component of the metric diverges towards the singularity while all other space components converge to zero. By redefining the time coordinate as $t:=-\ln τ$ the metric can be brought to the form required in Lemma~\ref{lem:expanding_singularity}:  since the off-diagonal components $g_{θi}$ vanish, so \eqref{6.1_cond3} is satisfied.

\subsection{Further AVTD spacetimes}
 \label{s21X16.2}

Using Fuchsian methods, asymptotically Kasner-like spacetimes without symmetries have been constructed in the presence of various matter fields
 and in vacuum (either for spacetime dimension  higher than 10, or the ones described in section \ref{s20X17.1} which exist for a restricted set of asymptotic data).
The constructions generally start by defining a reduced evolution system, the ``velocity term dominated" (VTD) or ``Kasner-like" system,
which does not include spatial derivatives, and then using Fuchs-type theorems to show that solutions of the full Einstein equations which approach these exist. These theorems guarantee a convergence which is uniform on compact subsets, as required by Lemma~\ref{lem:expanding_singularity}. Assuming that the spatial manifold is compact, the only things left to verify are the conditions \eqref{6.1_cond12}, \eqref{6.1_cond3}

In the case of stiff fluid or scalar field matter there is no expanding direction and Lemma \ref{lem:expanding_singularity} does not apply~\cite{AnderssonRendall,DHRW}.

In the case of $\geq10$ dimensional vacuum, there is at least one expanding direction~\cite{DHRW}. These solutions can be constructed using the same approach as those in section \ref{s20X17.1}, leading to the same behavior of the $\Ni{^a_i}$ and therefore also satisfy $\eqref{6.1_cond3}$. Lemma \ref{lem:expanding_singularity} is applicable if the spatial manifold is compact.

There are various results on general (non-Gowdy) $T^2$ symmetric spacetimes, the most general of which assumes the so-called ``half-polarization" condition~\cite{Clausen2007,Ames2013}. These have one expanding direction, satisfy \eqref{6.1_cond3}, and are spatially compact.

\subsection{``Mixmaster" Bianchi IX spacetimes}\label{sec:bianchiix}

The Bianchi models are homogeneous, but generically anisotropic, spacetimes, which are divided into types according to the structure constants of their Killing vector fields.

As shown by Ringström~\cite{Ringstrom2000}, generic solutions of type IX have at least
three $α$-limit points on the ``Kasner circle", \ie they approach at least three different Kasner metrics arbitrarily closely as the singularity is approached. It is conjectured that generic $α$-limit sets contain an infinite number of points on the Kasner circle, and that the dynamics approaches that of the discrete ``Kasner map", which shows chaotic behavior~\cite{Heinzle2009}. One would naively expect to be able to choose a sequence of times and directions as in Remark \ref{rem:sequence}. However, it has been pointed out to us by Hans Ringstr\"om (private communication) that this expectation is incorrect, and that with some work one can infer from~\cite{Ringstrom2000} that the space-diameter of the surfaces of homogeneity in all Bianchi IX vacuum models approaches zero as the singularity is approached.


\subsection{Generic spacelike singularities in the context of the BKL conjecture}\label{sec:genericbkl}

The BKL conjecture states, roughly speaking, that generic spacelike singularities behave at each spatial point as a ``Mixmaster" Bianchi IX solution.

There are various heuristic arguments supporting this conjecture. Using the so-called ``cosmological billiards formalism", as described in~\cite{Damour2003}, the metric takes the form \eqref{10II15.1}, with the same behavior of the $\Ni{^a_i}$ as in \eqref{10II15.2}, and with the $β^a(τ,x)$ now \emph{not} showing linear behavior in $τ$, but rather a sequence of approximately linear phases (so-called Kasner epochs) connected by ``bounces" off  increasingly sharp potential walls. The $β^a$ are expected to be unbounded towards the singularity ($τ \to \infty$), but they might well be bounded from below, and therefore it is not clear (and perhaps unlikely, given the Bianchi IX result mentioned above) whether one can choose a sequence of times and directions as described in Remark \ref{rem:sequence}.

It is expected that such solutions will develop particle horizons (this is sometimes referred to as ``asymptotic silence")\cite{heinzle2009_attractor}. This would imply that the second option in Lemma \ref{lem:expanding_singularity} is fulfilled, \ie $I^-(p, M)\subset (0,t_σ)\times K_p$ for all points $p$ sufficiently close to the singularity and some compact set $K_p$.
Whether or not this can be used to infer that these spacetimes are $C^0$ inextendible remains to be seen.

It should be emphasized that the arguments of~\cite{Damour2003} are heuristic.

\bigskip
\noindent{\sc Acknowledgements}
We are grateful to the Erwin Schr\"odinger International Institute for Mathematics and Physics for hospitality during part of work on this project. The research of PTC was supported in part by the Polish National Center of Science (NCN) under grant 2016/21/B/ST1/00940. PK was supported by a uni:docs grant of the University of Vienna. Useful discussions with Greg Galloway are acknowledged.

\bibliographystyle{iopart-num-arxivurl}
\bibliography{./bib}
\end{document}